\documentclass[submission]{eptcs}
\providecommand{\event}{Fics 2012} 
\usepackage{amssymb}
\usepackage{amsmath}
\usepackage{graphics}
\newcommand{\sqlt}{\sqsubseteq}

\DeclareMathOperator{\Fix}{\nu}
\newcommand{\ioi}{\Leftrightarrow}

\newcommand{\der}[1]{\ensuremath{\stackrel{#1}{\longrightarrow}}}
\newcommand{\sder}[1]{\ensuremath{\stackrel{#1}{\rightarrow}}}

\newcommand{\fun}{\rightarrow}

\newcommand{\ol}[1]{\overline{ #1}}
\newcommand{\Act}{{\bf A}}
\newcommand{\Proc}{{\bf P}}
\newcommand{\Psys}{P}
\newcommand{\F}{{\cal F}}
\newcommand{\V}{{\cal V}}
\newcommand{\G}{{\cal G}}
 
\newcommand{\Po}{{\cal P}}
\newcommand{\M}{{\cal M}}
\newcommand{\X}{{\cal X}}

\newcommand\ff{f\!\!f} 
\newcommand\tr{t\!t}
\newcommand\lb {[\![}
\newcommand\rb{]\!]}

\newcommand{\must}[1]{[ #1 ]}
\newcommand{\may}[1]{\langle #1 \rangle}

\newcommand{\smay}[1]{\langle\cdot #1 \cdot\rangle}
\newcommand{\smust}[1]{[ \cdot #1 \cdot ]}
\newcommand{\sem}[1]{\relax\ifmmode \lb #1 \rb \else $\lb #1 \rb$ \fi}
\newcommand{\semp}{\sem}

\newtheorem{theorem}{Theorem}[section]
\newtheorem{lemma}[theorem]{Lemma}

\newtheorem{corollary}[theorem]{Corollary}
\newtheorem{definition}[theorem]{Definition}

\newenvironment{proof}[1][Proof]{\begin{trivlist}
\item[\hskip \labelsep {\bfseries #1}]}{\end{trivlist}}
\newenvironment{example}[1][Example]{\begin{trivlist}
\item[\hskip \labelsep {\bfseries #1}]}{\end{trivlist}}

\title{Characteristic Formulae for Relations with Nested Fixed
  Points\thanks{Supported by the project Processes and Modal Logics'
    (project nr.~100048021) of the Icelandic Research Fund.} }
\author{Luca Aceto\quad Anna Ing\'olfsd\'ottir\thanks{Supported by the
    VELUX visiting professorship funded by the VILLUM FOUNDATION.}
  \institute{ICE-TCS, School of Computer Science\\Reykjavik University \\
    Reykjavik, Iceland} \email{\quad $\{$luca,annai$\}$@ru.is}}

\def\titlerunning{Fixed points}

\def\authorrunning{Aceto \&
    Ing\'olfsd\'ottir}
\begin{document}
\maketitle
\begin{abstract}
  A general framework for the connection between characteristic
  formulae and behavioral semantics is described in
  \cite{AILS2011}. This approach does not suitably cover semantics
  defined by nested fixed points, such as the $n$-nested simulation
  semantics for $n$ greater than $2$. In this study we address this
  deficiency and give a description of nested fixed points that
  extends the approach for single fixed points in an intuitive and
  comprehensive way.
\end{abstract}          

\section{Introduction}
In process theory it has become a standard practice to describe
behavioural semantics in terms of equivalences or preorders. A wealth
of such relations has been classified by van Glabbeek in his linear
time/branching time spectrum~\cite{vG2001}. Branching-time behavioural
semantics are often defined as largest fixed points of monotonic
functions over the complete lattice of binary relations over
processes.

In \cite{AILS2011} we give a general framework to reason about how
this type of behavioral semantics can be characterized by a modal
logic equipped with a greatest fixed point operator, or more precisely
by characteristic formulae expressed in such a logic. In that
reference we show that a behavioural relation that is derived as a
greatest fixed point of a function of relations over processes is
given by the greatest fixed point of the semantic interpretation of a
logical declaration that expresses the function in a formal sense that
is defined in present paper.  Roughly speaking if a logical
declaration describes a monotonic function over a complete lattice
then its fixed point describes exactly the fixed point of the
function. In \cite{AILS2011} preorders and equivalences such as simulation preorder
and bisimulation equivalence are characterized following this approach
in a simple and constructive way. However, when the definition of a
behavioural relation involves nested fixed points, i.~e.~when the
monotonic function that defines the relation takes another fixed point
as an argument, things get more complicated.  The framework offered in
\cite{AILS2011} only deals with nesting on two levels and in a rather
clumsy and unintuitive way. Furthermore it does not extend naturally
to deeper nesting, like for the $n$-nested simulations for $n>2$. In
this study we address this deficiency and define a logical framework
in which relations obtained as a chain of nested fixed points of
monotonic functions can be characterized following general
principles. This extends the approach for single fixed points in an
intuitive and comprehensive way.

As the applications we present in the paper only deal with nesting of
greatest fixed points, this study only focuses on greatest fixed
points. However it is straightforward to extend it to deal with
alternating nesting of both least and greatest fixed points.  We also
believe that our approach gives some idea about how fixed point
theories in different domains can be compared in a structured way.

The remainder of the paper is organized as follows. Section~2 presents
some background on fixed points of monotone functions. Section~3
briefly introduces the model of labelled transition systems and some
results on behavioural relations defined as greatest fixed points of
monotonic functions over binary relations. The logic we shall use to
define characteristic formulae in a uniform fashion is discussed in
Section~4. The key notion of a declaration expressing a monotone
function is also given in that section. Section~5 is devoted to an
application of our framework to the logical characterization of the
family of nested simulation semantics.

\section{Posets, monotone functions and fixed points}
In this section we introduce some basic concepts we need in the paper.
\begin{definition}\quad
\begin{itemize}
\item 
  A \emph{partially ordered set}, or \emph{poset}, $(A,\sqsubseteq_A)$
  (usually referred to simply as $A$) consists of a set $A$ and a
  partial order $\sqlt_A$ over it. 
\item If $A$ is a poset and $M\subseteq A$, then $a\in A$ is an
  \emph{upper bound} for $M$ if $m\sqlt_ Aa$ for all $m\in M$. $a$ is a
  \emph{least upper bound} (lub) for $M$ if it is an upper bound for
  $M$ and if whenever $b$ is an upper bound for $M$ then $a\sqlt_ Ab$.
\item A poset $A$ is a \emph{complete lattice} if the $\textit{lub}$
  for $M$ exists for all $M\subseteq A$.
\item For posets $A$ and $B$, a function $\phi: A \fun B$ is
  monotone if it is order preserving; it is an \emph{isomorphism} if it is
  bijective and both $\phi$ and its inverse $\phi^{-1}$ are monotone.
  We let $A\fun_{mono} B$ denote the set of monotone functions from $A$ to
  $B$.
\item If $A$ is a poset and $f\in A\fun_{mono} A$, then $x \in A$ is a fixed
  point of $f$ if $f(x)=x$.  We write $\nu\!f$ (or $\nu x.f(x)$) for
  the greatest fixed point  of $f$ if it exists.
\item If $A$ and $B$ are posets, $f\in A\fun_{mono} A$ and $\phi\in
  A\fun_{mono} B$ is an isomorphism then we define $\phi^*f:B\fun B$ as
  $\phi^*f=\phi\circ f\circ\phi^{-1}$.
\end{itemize}
\end{definition}
Note that the $\mathit{lub}$ of a subset of a poset $A$ is unique if it exists
and the same holds for greatest fixed points of monotone functions over
posets. It is well known, that if $A$ and $B$ are posets/complete lattices
and $I$ is some set, then the Cartesian product $A\times B$ and the
function space $I\fun A$ are a posets/complete lattices under the
pointwise ordering. The following theorem is due to Tarski.
\begin{theorem}[\cite{Ta55}]\label{tarski}
  If $A$ is a complete lattice and $f\in A\fun_{mono} A$, then $f$ has
  a unique greatest fixed point.
\end{theorem}
The theorem below is proved in \cite{AILS2011} and is the key to the
general theory we present in this paper.
\begin{theorem} \label{thm:isomorph}Let $A$ and $B$ be posets, $f\in
  A\fun_{mono} A$ and $\phi: A\fun B$ be an isomorphism. Then $\nu f$
  exists iff $\nu (\phi^*f)$ exists. If these fixed points exist then
  $\phi(\nu f)= \nu(\phi^*f)$.  
\end{theorem}
\section{Labelled transition systems and behavioural relations}\label{sec:LTS}
It has become standard practice to describe behavioural semantics of
processes by means of a \emph{labelled transition system} as defined
below.  
 \begin{definition}[\cite{Ke76}]
   A \emph{labelled transition system (LTS)} is a triple
   $\Psys =(\Proc,\Act, \sder{})$ where
   \begin{itemize}
   \item $\Act$ is a finite set (of actions),
  \item $\Proc$ is a finite set (of processes), and
  \item $\sder{}\subseteq \Proc\times\Act\times\Proc$ is a transition relation.
  \end{itemize}
\end{definition}
As usual, we write $p\der{a}p'$ for $(p,a,p')\in \sder{}$. Throughout
this paper we assume that the set $\Act$ is fixed.

As LTSs are in general to concrete, processes are compared by
preorders or equivalences. These are often obtained as the greatest
fixed points to monotone endofunctions on the complete lattice
$\Po(\Proc \times \Proc)$. We will show some example of such functions
but first we state and  prove some properties.
\begin{definition}
  If $\F\in\Po(\Proc \times \Proc)\fun_{mono}\Po(\Proc \times \Proc)$
  and $A\in\Po(\Proc \times \Proc)$, we define
\begin{itemize}
\item
$\tilde{\F}:S \mapsto (\F(S^{-1}))^{-1}$, and 
\item
$\F\cap A: S\mapsto\F(S)\cap A$.
\end{itemize}
\end{definition}
The following lemma  will be applied below.
\begin{lemma} \label{L:inverse} Let $\F\in\Po(\Proc \times
  \Proc)\fun_{mono}\Po(\Proc \times \Proc)$ and $A\in\Po(\Proc \times
  \Proc)$.  Then 
\begin{itemize}
\item
$\tilde{\F},\F\cap A\in\Po(\Proc \times
  \Proc)\fun_{mono}\Po(\Proc \times \Proc)$,
\item
$\nu \tilde{\F}=(\nu\F)^{-1}$ and
\item $\widetilde{\F\cap A}=\tilde{\F}\cap A^{-1}$.
\end{itemize}
\end{lemma}
\begin{proof}
  The first two statements are proved in \cite{AILS2011}. To prove the third
  one we proceed follows:
\[
(\widetilde{\F\cap A})(S)=((\F\cap A)(S^{-1}))^{-1}=(\F(S^{-1}))^{-1}\cap A^{-1}=(\tilde{\F}\cap A^{-1})(S).
\]
\end{proof}
We will complete this section by giving some examples of endofunction that
define some standard behavioural preorders and equivalences
\cite{vG2001,AILS2007}.
\begin{definition}
  Let $\F:\Po(\Proc \times \Proc)\fun \Po(\Proc \times \Proc)$ be
  defined as follows: 
\[
(p,q)\in\F(S)\mbox{ iff }\forall
a\in\Act,p'\in\Proc. p\der{a}p'\Rightarrow \exists q'\in\Proc
\,.q\der{a}q'\land(p',q')\in S.
\]
\end{definition}
It is easy to check that $\F$ is monotonic and therefore it has a
greatest fixed point.
\begin{definition}
We define:
\begin{itemize}
\item $\F_{sim}=\F$ and $\sqlt_{sim}=\nu\F_{sim}$ (simulation
  preorder),
\item $\F_{opsim}=\tilde{\F}$ and $\sqlt_{opsim}=\nu\F_{opsim}$ (inverse
  simulation preorder),
\item $\sim_{sim}=\sqlt_{sim}\cap \sqlt_{opsim}$ (simulation equivalence) and
\item $\F_{bisim}=\F_{sim}\cap\F_{opsim}$ and
  $\sim_{bisim}=\nu\F_{bisim}$ (bisimulation equivalence).
\end{itemize}
\end{definition}
\section{Equational modal $\nu$-calculi with nested fixed-points}
\label{sect:hennmil}
In this section we introduce variants of the standard equational modal
$\mu$-calculus \cite{Koz83}. Like in \cite{Larsen1990} these variants
only allow for nested fixed points, i.~e.~where the logical languages
form a hierarchy where fixed points in a language on one level are
allowed as constants in the logic on the level above. Our approach,
however, differs from the original one in the sense that the
fixed-point operator is explicit in the syntax and can therefore be
used in logical expressions. In this study we only focus on greatest
fixed points (which explains the title of this section) but the
framework can easily be extended to involve nesting of both greatest
and least fixed points. The logical languages we introduce depend on
the implicitly assumed fixed finite set $\Act$.

Our basic logic $\M$ is the standard Hennessy-Milner Logic (HML)
\cite{HM85} without variables.  This logic is generated by
$\Sigma=(\Sigma_0,\Sigma_1, \Sigma_2)$ where $\Sigma_0=\{\tr,\ff\}$
are the constants or the operators of arity $0$,
$\Sigma_1=\{\may{a},\must{a},a\in\Act\}$ are the operators of arity
$1$, and $\Sigma_2=\{\land,\lor\}$ are the operators of arity
$2$.

The formulae in $\M$ are interpreted over an LTS $(\Proc, \Act,
\sder{})$ as the set of elements from $\Proc$ that satisfy
them. Satisfaction is determined by a semantic function that is
defined below.  For $M\subseteq\Proc$ we let
$\smay{a}M=\{p\in\Proc\mid\exists q\in M.p\der{a}q\}$, and $\smust{a}M
=\overline{\smay{a}\overline{M}}$ where $\overline{M}$ is the
complement of the set $M$.
\begin{definition}\label{def:semp}\quad
The semantic function  $\M\sem{\,\,\,}$ is defined as follows:
\begin{enumerate}
\item $\M\semp{\tr}=\Proc,\,\M\semp{\ff}=\emptyset$,
\item
  $\M\semp{F_1\land F_2}= \M\semp{F_1}\cap\M\semp{F_2},\,
  \M\semp{F_1\lor F_2}= \M\semp{F_1}\cup\M\semp{F_2}$,
\item
  $\M\semp{\may{a}F}=\smay{a}\M\semp{F}, \,
  \M\semp{\must{a}F}=\smust{a}\M\semp{F}$.
\end{enumerate}
\end{definition}
The logic $\V$ is the standard Hennessy-Milner logic with variables
that was introduced in \cite{Larsen1990}. It assumes a finite index
set $I$ and an $I$-indexed set of variables $\X$. In what remains of
this paper we assume a fixed pair of such $I$ and $\X$, unless stated
otherwise.

As the elements of $\V$ typically contain variables, they have to
be interpreted with respect to a variable \emph{interpretation}
$\sigma\in \Po(\Proc)^{I}$ that associates to each $i\in I$ the set of
processes in $\Proc$ that are assumed to satisfy the variable $X_i$.
The semantic function $\V\semp{\,\,\,}$ in this case takes a formula
$F$ and a $\sigma\in\Po(\Proc)^{I}$ and delivers an element of
$\Po(\Proc)$.
\begin{definition}\label{def:V}
The semantic function $\V\semp{\,\,\,}$ is defined as follows:
\begin{enumerate}
\item $\V\semp{F}\sigma=\M\semp{F}$ if $F\in\Sigma_0$,
\item 
  $\V\semp{X_i}\sigma=\sigma(i)$, $i\in I$,
\item $\V\semp{F_1\land F_2}\sigma=
  \V\semp{F_1}\sigma\cap\V\semp{F_2}\sigma,\,\,
  \V\semp{F_1\lor F_2}\sigma=
  \V\semp{F_1}\sigma\cup\V\semp{F_2}\sigma$,
\item
  $\V\semp{\may{a}F}\sigma=\smay{a}\V\semp{F}\sigma, \,
  \V\semp{\must{a}F}\sigma=\smust{a}\V\semp{F}\sigma$. 
\end{enumerate}
\end{definition}
In \cite{Larsen1990} the meaning of the variables in the logic $\V$ is
defined by means of a declaration, or a function
$D:I\fun\V$. Intuitively the syntactic function generates a monotonic
endofunction $\V\semp{D}$ over $\Po(\Proc)^I$ defined by
$(\V\semp{D})(i)=\V\semp{D(i)}$ for all $i\in I$. By Theorem
\ref{tarski}, $\V\semp{D}$ has a unique largest fixed point
$\nu\V\semp{D}\in \Po(\Proc)^I$ that can be used to give the semantics
for the variables and the formulae that contain those in the logic
$\V$.  We can then use this to extend the logic $\M$ with $\{\nu D(i)|
i\in I\}$ as constants  interpreted as
$\{\nu\V\semp{D}(i)|i\in I\}$. By this we get a logic $\M'$ that is
generated by $\Sigma'=(\Sigma_0\cup \{\nu D(i)|i\in I\}, \Sigma_2,
\Sigma_3)$. Then this procedure can be repeated for another
declaration that possibly depends on $\nu D$ as a constant and with
$\M'$ as the basic logic. The following example shows how this
construction works.
\begin{example}
  Let $I=\{1\}$, $\X=\{X_1\}$ and $\Act=\{a,b \}$ and let the property
  ``invariantly $\may{a}\tr$'' be defined as the greatest fixed point
  corresponding to the declaration $D_0$ defined as
  $D_0(1)=\may{a}\tr\land\must{a}X_1\land\must{b}X_1$.  To interpret
  this we define $\M=\M_0$ and $\V_0=\V$ where $\M$ and $\V$ have the
  meaning described above.  The derived semantic function
  $\V_0\sem{D_0}:\Po(\Proc)^{\{1\}}\fun\Po(\Proc)^{\{1\}}$ is easily
  shown to be monotonic and has the greatest fixed point
  $\nu\V_0\sem{D_0}\in\Po(\Proc)^{\{1\}}$. Now we define $\M_1$ as the
  extension of $\M_0$ that is generated by $\Sigma^1=(\{\tr,\ff,\nu D_0(1)\},
  \Sigma_1,\Sigma_2)$, i.e. has $\nu D_0(1)$ as a constant that is
  interpreted as $\nu\V_0\sem{D_0}(1)$, i.e.~$\M_1\sem{\nu
    D_0(1)}=\nu\V_0\sem{D_0}(1)$.

  Next let us assume that we have the declaration
  $D_1:{\{1\}}\fun\V_1$ where $\V_1$ is the variable logic generated
  by $(\{\tr,\ff,\nu D_0(1),X_1\}, \Sigma_2, \Sigma_3)$ and $D_1$ is
  defined as $D_1(1)=\may{b}\nu D_0(1)\land\must{b}X_1$ . As before
  the declaration is interpreted over $\Po(\Proc)^{\{1\}}$ but using
  $\M_1\sem{\,\,\,}$ to interpret the constant $\nu D_0(1)$. Again
  $D_1$ is interpreted by using $\V_1\sem{\,\,\,}$ which leads to a
  monotonic endofunction $\V_1\sem{D_1}$ over $\Po(\Proc)^{\{1\}}$
  with a fixed point $\nu\V_1\sem{D_1}$. The logic $\M_2$ is now
  defined as the one generated by $\Sigma^2=(\{\tr,\ff,\nu D_1(1), \nu
  D_2(1)\}, \Sigma_2,\Sigma_3)$ where $\M_0\sem{\,\,\,}$ and
  $\M_1\sem{\,\,\,}$ are used to define the meaning of $\nu D_1(1)$
  and $\nu D_2(1)$ respectively.
\end{example}
We will now generalize this procedure and define our hierarchy of
nested fixed point logics, derived from a sequence of nested
declarations $D_j, j=1,2,\ldots, N$, i.e. where for each $n<N$,
$D_{n+1}$ is allowed to depend on the constants $\tr,\ff$ and $\nu
D_j(i)$ for $j\leq n$ and $i\in I$. In the definition we assume a
finite index set $I$ and an $I$-indexed variable set $\X$.  We use the
notation $\G(\Sigma_0)$ for the logic generated by
$(\Sigma_0,\Sigma_1,\Sigma_2)$ and $\G_I(\Sigma_0)$ for the logic
generated by $(\Sigma_0\cup\X,\Sigma_1,\Sigma_2)$.

\begin{definition}\quad
\begin{itemize}
\item Define 
\begin{itemize}
\item
$\Sigma_0^0=\{\tr,\ff\}$, 
\item $\M_0=\G(\Sigma_0^0)$ and
\item $\V_{0}=\G_I(\Sigma^0_0)$.
\end{itemize}
\item For $n\geq 1$, if $D_{n}:I\fun\V_{n}$, define
\begin{itemize}
 \item $\Sigma^{n+1}_0=\Sigma^{n}_0\cup\{\nu D_n(i)| i\in I \}$,
\item $\M_{n+1}=\G(\Sigma^{n+1}_0)$ and
\item $\V_{n+1}=\G_I(\Sigma^{n+1}_0)$.
\end{itemize}
\end{itemize}
\end{definition}

To define the semantic functions associated with these logics we need the
following lemma.
\begin{lemma}\label{L:logic}\quad
  Assume that $\M=\G(C)$ and $\V=\G_I(C)$ for some set
  of constants $C$ where $\M\sem{c}$ is well defined for
  all $c\in C$.  Then for all $D: I\fun\V$, the derived semantic
  function $\V\sem{D}$ defined by
\[
\forall
  i\in I.(\V\semp{D}\sigma)(i)=\V\semp{D(i)}\sigma
 \]
 is in $\Po(\Proc)^{I}\fun_{mono}\Po(\Proc)^{I}$ and hence, by Theorem
 \ref{tarski}, $\nu\V\sem{D}\in \Po(\Proc)^{I}$
 exists.
\end{lemma}

Now we are ready to define the semantic functions
for $\M_{n}$ and $\V_n$ for all $n\geq 0$.
\begin{definition}\quad
\begin{itemize}
\item $\M_{0}=\M$ and $\V_{0}=\V $ as defined in Definition
  \ref{def:semp} and \ref{def:V} respectively.

\item For $n\geq 0$ the semantic functions for $\M_{n+1}$ is  defined as follows:
\begin{enumerate}
\item 
$\M_{n+1}\semp{F}=\M_{n}\semp{F}$ if $F\in\Sigma^n_0$,
\item
$\M_{n+1}\semp{(\nu D_{n})(i)}=\nu \V_n\semp{D_{n}}(i)$ for $i\in I$,
\item
  $\M_{n+1}\semp{F_1\land F_2}= \M_{n+1}\semp{F_1}\cap\M_{n+1}\semp{F_2},\,
  \M_{n+1}\semp{F_1\lor F_2}= \M_{n+1}\semp{F_1}\cup\M_{n+1}\semp{F_2}$,
\item
  $\M_{n+1}\semp{\may{a}F}=\smay{a}\M_{n+1}\semp{F}, \,
  \M_{n+1}\semp{\must{a}F}=\smust{a}\M_{n+1}\semp{F}$.
\end{enumerate}
\item For $n\geq 0$ the semantic function for $\V_{n+1}$ is defined as follows:
\begin{enumerate}
\item $\V_{n+1}\semp{F}\sigma=\M_{n+1}\semp{F}$ if $F\in\Sigma^0_n$,
\item 
  $\V_{n+1}\semp{X_i}\sigma=\sigma(i)$, $i\in I$,
\item $\V_{n+1}\semp{F_1\land F_2}\sigma=
  \V_{n+1}\semp{F_1}\sigma\cap\V_{n+1}\semp{F_2}\sigma,\,\,
  \V_{n+1}\semp{F_1\lor F_2}\sigma=
  \V_{n+1}\semp{F_1}\sigma\cup\V_{n+1}\semp{F_2}\sigma$,
\item
  $\V_{n+1}\semp{\may{a}F}\sigma=\smay{a}\V_{n+1}\semp{F}\sigma, \,
  \V_{n+1}\semp{\must{a}F}\sigma=\smust{a}\V_{n+1}\semp{F}\sigma$. 
\end{enumerate}
\end{itemize}

\end{definition}

\subsection{Characteristic Formulae by means of
  Declarations}\label{sect:chargfp}
The aim of this section is to show how each process $p\in\Proc$
can be characterized up to a binary relation $\bowtie$ over processes (such
as an equivalence or a preorder) by a single formula, the so called
characteristic formula for $p$ up to $\bowtie$.  

To achieve this, we take $I=\Proc$ in the definitions in the
previous section.  A declaration $D$ for a variable logic $\V$
assigns exactly one formula $D(p)$ from $\V$ to each process
$p\in\Proc$.  We have seen that each such function induces an
endofunction
$\V\semp{D}\in\Po(\Proc)^\Proc\fun_{mono}\Po(\Proc)^\Proc$ and
therefore $\V\semp{D}$ exists. This leads to the following
definition:

\begin{definition}\label{D:char}
  A declaration $D$ for the logic $\V$ characterizes
  $\bowtie\subseteq \Proc\times \Proc$ iff for each $p,q\in \Proc$,
$$
(p,q)\in\bowtie\text{ iff }q \in (\nu\V\semp{D})(p).
$$
\end{definition}
In what follows, we will describe how we can devise a
characterizing declaration for a relation that is obtained as a fixed
point, or a sequence of nested fixed points of monotone
endofunctions, which can be expressed in the logic.  In order to define
this precisely we use the notation introduced in Definition
\ref{Def:sigmaS} below.
\begin{definition}\label{Def:sigmaS}
If $S\subseteq\Proc\times\Proc$ we define the variable interpretation
$\sigma_S\in\Po(\Proc)^{\Proc}$ associated to $S$ by
\[
\sigma_S(p)=\{q\in\Proc \mid (p,q)\in S\}, \textrm{ for each $p\in \Proc$}.
\]
\end{definition}
Thus $\sigma_S$ assigns to $p$ all those processes $q$ that are related to it
via $S$.
\begin{definition} \label{def:express} A declaration $D$ for $\V$
  \emph{expresses} a monotone endofunction $\F$ on $\Po(\Proc\times
  \Proc)$ when
\[
(p,q)\in\F(S) \text{ iff }q\in\V\semp{D(p)}\sigma_S=(\V\semp{D}\sigma_S)(p) ,
\]
   for every relation $S \subseteq \Proc \times \Proc$ and every
   $p, q  \in \Proc$.
\end{definition}
We need the following to prove our main result.
\begin{definition}\label{def:isophi}
Let $\Phi:\Po(\Proc\times \Proc)\fun\Po(\Proc)^{\Proc}$ be defined by
$\Phi(S)=\sigma_S$. 
\end{definition}
\begin{lemma}\label{L:x}\quad
\begin{itemize}
\item
 $\Phi:\Po(\Proc\times \Proc)\fun\Po(\Proc)^{\Proc}$ is an
 isomorphism.
\item If $A_1,A_2\in\Po(\Proc\times \Proc)$ and
  $\F_1,\F_2\in\Po(\Proc\times \Proc)\fun_{mono}\Po(\Proc\times
  \Proc)$ then 
\begin{itemize}
\item
$\Phi(A_1\cap A_2)=\Phi(A_1)\cap\Phi(A_2)$, 
\item
$\Phi^*(\F_1\cap
  A_1)=\Phi^*(\F_1)\cap\Phi(A_1)$ and
\item $\Phi^*(\F_1\cap\F_2)=\Phi^*(\F_1)\cap\Phi^*(\F_2)$.
\end{itemize}
\end{itemize}
 \end{lemma}
\begin{proof}
  The first part is proved in \cite{AILS2011} whereas the second part
  follows directly from the definition of $\Phi$.
\end{proof}
\begin{corollary}\label{L:Phiiso}
  Assume that $D\in\Proc\fun\V$ and $\F\in \Po(\Proc\times
  \Proc)\fun_{mono}\Po(\Proc\times \Proc)$. Then 
\[
D \mbox{ expresses }\F\,\mbox{ iff }\,\Phi^*(\F)=\V\semp{D}
\,\mbox{ iff }\,D \mbox{ characterizes }\nu\F.
\]
 \end{corollary}
\section{Applications}

Following the approach in \cite{AILS2011}, we define declarations $D$
and $\tilde{D}$ that express the functions $\F$ and $\tilde{\F}$ that
were defined in Section \ref{sec:LTS}. 
\begin{definition}Let
\begin{itemize}
\item Let $D: p\mapsto\bigwedge_{a \in \Act} \bigwedge_{p' \in
    \Proc.\, p \der{a} p'}\may{a}X_{p'} $ and 
\item $\tilde{D}:
  p\mapsto\bigwedge_{a\in \Act} \must{a}\bigvee_{p' \in \Proc.\, p
    \der{a} p'}X_{p'} $.
\end{itemize}
\end{definition}
\newpage
From \cite{AILS2011} we have:
\begin{lemma}\label{lemma:D-F}\quad
\begin{itemize}
\item
  $D$ expresses $\F$ and characterizes $\nu\F$,  and 
\item
$\tilde{D}$
  expresses $\tilde{\F}$ and characterizes $\nu\tilde{\F}$.
\end{itemize}
\end{lemma}
Now we recall from \cite{AILS2011} the declarations that characterize simulation
equivalence and bisimulation equivalence.
\begin{definition}
Define
$D_{bisim}=D_{sim}\land D_{opsim}$ and
$D_{simeq}=\nu D_{sim}\land \nu D_{opsim}$.
\end{definition}
\begin{lemma}\quad
 $D_{bisim}$ characterizes $\sim_{bisim}$ and $D_{simeq}$ characterizes $\sim_{sim}$. 
\end{lemma}
\begin{proof}
$D_{bisim}$ does not contain nested fixed points and can
  therefore be interpreted directly over $\V_0=\V$. 
Now we proceed as follows:
\[
\Phi^*(\F_{bisim})=\Phi^*(\F_{sim})\cap\Phi^*(\F_{opsim})=
\V\sem{D_{sim}}\cap\V\sem{D_{opsim}}=\V\sem{D_{sim}\land D_{opsim}}=\V\sem{D_{bisim}}.
\]
To interpret $D_{simeq}$ we define $\Sigma_1=\{\tr,\ff\}\cup\{\nu
D_{sim}(p)|p\in\Proc\}$ and $\Sigma_2=\Sigma_1\cup\{\nu
D_{opsim}(p)|p\in\Proc\}$ and let $\M_0,\M_1,\M_2$ and
$\V_0,\V_1$ be defined as before. Then $D_{simeq}:\Proc\fun\V_1$.
If we let $\F_{simeq}=\nu\F_{sim}\cap\nu\F_{opsim}$, we get 
\[
\begin{array}{l}
\Phi^*(\F_{simeq})=\Phi(\nu\F_{sim})\cap\Phi(\nu\F_{opsim)}=
\nu\V_1\sem{D_{sim}}\cap\nu\V_1\sem{D_{opsim}}=\\\\
\M_2\sem{\nu D_{sim}}\cap\M_2\sem{\nu D_{opsim}}=
\M_2\sem{\nu D_{sim}\land\nu D_{opsim}}=\V_1\sem{D_{simeq}}.
\end{array}
\]
The result now follows from Cor.~\ref{L:Phiiso}.
\end{proof}

Next we define the nested simulation preorders introduced in
\cite{GV92} by using the function $\F$. These definition involve
nesting of fixed points and are defined recursively on the depth of
the nesting.  The $1$-nested simulation $\sqlt_{(1)sim}$ is just the
simulation preorder $\sqlt_{sim}$ as defined in Section \ref{sec:LTS}
and the function $\F_{(1)sim}$ is therefore the function $\F$. As the
preorder $\sqlt_{(n+1)sim}$ depends on the inverse of the preorder
$\sqlt_{(n)sim}$, which we call $\sqlt_{(n)opsim}$, we simultaneously
define the nested simulations and their inverse in our recursive
definition.  The functions that define $\sqlt_{(n)sim}$ and
$\sqlt_{(n)opsim}$ are called $\F_{(n)sim}$ and $\F_{(n)opsim}$
respectively.

\begin{definition}[Nested simulations]\quad
\begin{enumerate}
\item $\F_{(1)sim}=\F$ and $\sqlt_{(1)sim}=\nu\F_{(1)sim}$, 
\item ${\F_{(1)opsim}}=\widetilde{\F}$ and
    $\sqlt_{(1)opsim}=\nu{\F_{(1)opsim}}$,
\item $\F_{(n+1)sim}=\F_{(1)sim}\cap\nu{\F_{(n)opsim}}$ and 
  $\sqlt_{(n+1)sim}=\nu\F_{(n+1)sim}$.
\item ${\F_{(n+1)opsim}}={\F_{(1)opsim}}\cap\nu\F_{(n)sim}$ and 
  $\sqlt_{(n+1)opsim}=\nu{\F_{(n+1)opsim}}$.
\end{enumerate} 
\end{definition}
We complete this note by defining a sequence of nested declarations
and prove that they characterize the sequence of $n$-nested
simulation preorders.
\begin{theorem}\quad
\begin{enumerate}
\item $D_{(1)sim}=D$\,  expresses\, $\F_{(1)sim}$\, and\, characterizes\,
    $\sqlt_{(1)sim}$,
  \item ${D_{(1)opsim}}=\widetilde{D}$\, expresses
    $\F_{(1)opsim}$\, and\, characterizes\,
    $\sqlt_{(1)opsim}$,
  \item $D_{(n+1)sim}=D_{(1)sim}\land\nu D_{(n)opsim}$\, expresses
    $\F_{(n+1)sim}$\, and\, characterizes\, $\sqlt_{(n+1)sim}$,
  \item ${D_{(n+1)opsim}}=D_{(1)opsim}\land\nu D_{(n)sim}$\,
    expresses\, ${\F_{(n+1)opsim}}$\, an\,d characterizes\,
    $\sqlt_{(n+1)opsim}$.
\end{enumerate}
\end{theorem}
\begin{proof}
  We prove the statements simultaneously by induction on $n$. First we
  note that $D_1, D_2,\ldots,$ where $D_{2i-2}=D_{(i)sim}$ and
  $D_{2i-1}=D_{(i)opsim}$ for $i\geq 1$ is a sequence of nested
  declarations.
 For the
  case $n=1$ we get from Lemma \ref{lemma:D-F} that
  $\Phi^*(\F_{(1)sim})=\V_{0}\semp{D_{(1)sim}}$ and
  $\Phi^*(\F_{(1)opsim})=\V_{1}\semp{D_{(1)opsim}}$.  Next assume that
  $\Phi^*(\F_{(n)sim})=\V_{2n-2}\semp{D_{(n)sim}}$ and
  $\Phi^*(\F_{(n)opsim})=\V_{2n-1}\semp{D_{(n)opsim}}$. To prove 3. we
  proceed as follows:
\[
\begin{array}{l}
\Phi^*(\F_{(n+1)sim})=\Phi^*(\F_{(1)sim})\cap\Phi(\nu\F_{(n)opsim})=
\V_{0}\semp{D_{(1)sim}}\cap\nu\V_{2n-2}\semp{D_{(n)opsim}}=
\\\V_{2n-2}\semp{D_{(1)sim}\land
  \nu D_{(n)opsim}}=\V_{2n}\semp{D_{(n+1)sim}}.
\end{array}
\]
Finally, to prove 4. we have:
\[
\begin{array}{l}
  \Phi^*(\F_{(n+1)opsim})=\Phi^*(\F_{(1)opsem})\cap\Phi(\nu\F_{(n)sim})=
  \V_{1}\semp{D_{(1)opsim}}\cap\nu\V_{2n-1}\semp{D_{(n)sim}}=\\
  \V_{2n-1}\semp{D_{(1)opsim}\land\nu D_{(n)sim}}=
\V_{2n+1}\semp{D_{(n+1)opsim}}.
\end{array}
\]
\end{proof}

\nocite{*}
\bibliographystyle{eptcs}
\bibliography{anna}
\end{document}